\pdfoutput=1
\documentclass[conference]{IEEEtran}

\usepackage{cite}
\usepackage[pdftex]{graphicx}
\usepackage{amsmath, amsfonts, amssymb, amsthm}
\usepackage{array}
\usepackage[caption=false,font=normalsize,labelfont=sf,textfont=sf]{subfig}
\usepackage[switch]{lineno}
\usepackage{bm}
\usepackage{color}
\usepackage{mathrsfs}
\usepackage{epstopdf}

\usepackage[ruled,linesnumbered]{algorithm2e}
\newtheorem{thm}{Theorem}

\newtheorem{ass}{Assumption}

\begin{document}

\title{Fast Convergence Algorithm for Analog Federated Learning}

\author{Shuhao~Xia$^\dagger$~,
                Jingyang~Zhu$^\dagger$, Yuhan~Yang$^\dagger$, Yong~Zhou$^\dagger$, 
        ~Yuanming~Shi$^\dagger$~% <-this % stops a space
        and Wei~Chen$^*$ \\
        $^\dagger$ School of Information Science and Technology, ShanghaiTech University, Shanghai 201210, China \\
        $^*$ Department of Electronic Engineering, BNRist, Tsinghua University, Beijing, 100084, China \\
        Email: xiashh@shanghaitech.edu.cn, shiym@shanghaitech.edu.cn
        
\thanks{}% <-this % stops a space
}

% make the title area
\maketitle

\begin{abstract}
In this paper, we consider federated learning (FL) over a noisy fading multiple access channel (MAC), where an edge server aggregates the local models transmitted by multiple end devices through over-the-air computation (AirComp). 
To realize efficient analog federated learning over wireless channels, we  propose an AirComp-based FedSplit algorithm, where a threshold-based device selection scheme is adopted to achieve reliable local model uploading. 
In particular, we analyze the performance of the proposed algorithm  and prove that the proposed algorithm linearly converges to the optimal solutions under the assumption that the objective function is strongly convex and smooth. We also characterize the robustness of proposed algorithm to the ill-conditioned problems, thereby achieving fast convergence rates and reducing communication rounds.  A  finite error bound is further provided to reveal the relationship between the convergence behavior and the channel fading and noise. Our algorithm is theoretically and experimentally verified to be much more robust to the ill-conditioned problems with faster convergence compared with other benchmark FL algorithms. 
\end{abstract}

%\begin{IEEEkeywords}
%Federated learning, over-the-air computation, convergence analysis.
%\end{IEEEkeywords}

\section{Introduction}
 As an emerging decentralized machine learning solution, federated learning (FL) has recently attracted considerable attention from both academia and industry. 
In FL, multiple devices with their local datasets collaboratively train a global model, where only local model updates instead of raw data are transmitted to the parameter server, thereby significantly reducing the bandwidth requirement and providing additional privacy protection \cite{9084352, 8664630}. 
Most of the existing studies on FL focused on  the reduction of the volume of model exchange without explicitly taking into account the impact of wireless channels.
However, FL has a wide range of applications in wireless networks, e.g., Internet of Things (IoT) \cite{10.1145/3298981}, autonomous driving \cite{9060868}. 
Therefore, it is essential to investigate the impact of the physical characteristics (e.g., channel distortion, noise) of the wireless medium on the convergence rate and  the optimality of FL algorithms.

%In this paper, we consider a \emph{wireless federated learning problem}, where multiple edge devices with the help of a central edge server collaboratively train a shared learning model in wireless communication systems.

Many digital communication based approaches have recently been proposed to facilitate FL in wireless networks \cite{9210812, 9194337, 8851249, 9170917}, where each edge device is assigned an orthogonal channel to upload its local model. 
In particular, the authors in \cite{9210812} formulated a joint resource allocation and user selection problem for FL, where a closed-form expression for the expected convergence rate of the FL was derived to establish an explicit relationship between the packet error rates and the FL performance. 
To reduce the total learning-and-communication latency, the authors in \cite{9194337} partitioned the learning task into multiple sub-tasks, which are allocated to different edge devices for parallel training. 
In addition, the authors in \cite{8851249, 9170917} further enhanced the communication efficiency of wireless federated learning systems by proposing efficient resource management mechanisms. 
As the orthogonal channels are required to enable concurrent local model uploading to avoid interference, the aforementioned studies may not be communication-efficient, especially when the number of edge devices is large. 

To support communication-efficient design, over-the-air computation (AirComp), as a promising analog multiple access scheme, is capable of achieving ultra-fast model aggregation for FL by allowing concurrent transmission from edge devices over the same frequency channel and exploiting the waveform superposition property of MAC \cite{8952884, 8870236, 9042352, 9076343}. 
In particular, the authors in \cite{8952884} proposed a joint device selection and beamforming design to accelerate the convergence of analog federated learning. 
In \cite{8870236}, a broadband analog aggregation scheme over MAC was proposed to reduce the communication latency. However, these studies did not analyze the convergence performance of FL algorithms. 
On the other hand, the authors in \cite{9042352} proposed a  distributed stochastic gradient descent (SGD) algorithm, in which each device transmits a sparse gradient estimate over MAC. In \cite{9076343}, the authors developed the gradient based multiple access (GBMA) algorithm, which is proved to achieve the same convergence rate as the centralized gradient descent (GD) algorithm in large-scale networks. 
Although the convergence analysis was provided, the aforementioned studies suffer from a high iteration complexity and high communication overhead under the ill-conditioned setting, which is well-known to be a performance-limiting factor. 
Very recently, the authors in \cite{2020arXiv200505238P} developed the FedSplit algorithm based on the operator splitting procedure to achieve fast convergence even in the ill-conditioned setting. 
However, both the convergence rate and the optimality of the FedSplit algorithm in wireless networks has not been studied, which motivates this work. 

In this paper, we consider a wireless FL problem over a noisy fading MAC. Due to the distortion and noise caused by MAC, the performance of FL algorithms over wireless channels are significantly degraded, especially under the ill-conditioned setting. To address these issues, we propose an AirComp-based FedSplit algorithm, in which the edge server aims to recover the aggregation of local models computed by the end devices via AirComp at each communication round.  We exploit a threshold-based device selection scheme to achieve reliable communication. For strongly convex and smooth local loss functions, we prove that the proposed algorithm can linearly converge to optimal points. Furthermore, we establish an error bound in term of the expected loss of the objective function to reveal the impact of channel fading and noise over convergence behavior. Finally, our theoretical results are well verified through numerical experiments under various parameter settings.

\emph{Notations}: All vectors are considered to be column vectors. We use boldface lowercase (uppercase) letters to represent vectors (matrices). We denote the identity matrix by $\mathbf{I}$, the set of real values by $\mathbb{R}$, the cardinality of set $A$ by $|A|$ and $\ell_2$-norm of vector $\bm{x}$ by $\|\bm{x}\|$. In addition, the function $f$ is defined to be $\ell$-strongly convex, if
\begin{align*}
        f(\bm{y}) \geq f(\bm{x})+\nabla f(\bm{x})^{\top}(\bm{y}-\bm{x})+\frac{\ell}{2}\|\bm{y}-\bm{x}\|^{2}
\end{align*}
for all $\bm{x}, \bm{y}$. Similarly, the function $f$ is defined to be $L$-smooth, if 
\begin{align*}
f(\bm{y}) \leq f(\bm{x})+\nabla f(\bm{x})^{\top}(\bm{y}-\bm{x})+\frac{L}{2}\|\bm{y}-\bm{x}\|^{2}
\end{align*}
for all $\bm{x}, \bm{y}$.

\section{System Model and Problem Formulation} \label{sec: system model}
\subsection{Federated Optimization}
We  consider a federated edge learning system consisting of $N$ single-antenna edge devices indexed by set $\mathcal{N} = \{1, 2, \ldots, N\}$ and a computing enabled edge server equipped with a single antenna, as illustrated in Fig. \ref{fig: system model}. Each device $n$ is associated with its own local dataset $\mathcal{D}_n$, and all edge devices collaboratively learn a shared global model by communicating with the edge server. 

In federated learning systems, the goal is to learn a shared global model by minimizing the sum of the devices' local loss function. Therefore, the problem can be formulated as the following consensus federated optimization problem:
\begin{align}\label{eq: problem}
%       \mathscr{P}: 
        \begin{array}{ll}
\underset{\bm{\theta}, \{\bm{\theta}_n\}_{n=1}^N }{\operatorname { minimize }} & F(\bm{\theta})\triangleq\sum_{n=1}^{N} f_{n}\left(\bm{\theta}\right) \\
\operatorname { subject~to } & \bm{\theta}_{n} = \bm{\theta}, \forall n \in \mathcal{N},
\end{array}
\end{align}
where $ \bm{\theta} \in \mathbb{R}^d$ is the global model with dimension $d$. For each device $n$, $\bm{\theta}_n \in \mathbb{R}^d $ is the local model and $f_n$ is the local loss function defined by the learning task and the local dataset $\mathcal{D}_n$. 

\begin{figure}[tbp]
        \centering
        \includegraphics[width=1\linewidth]{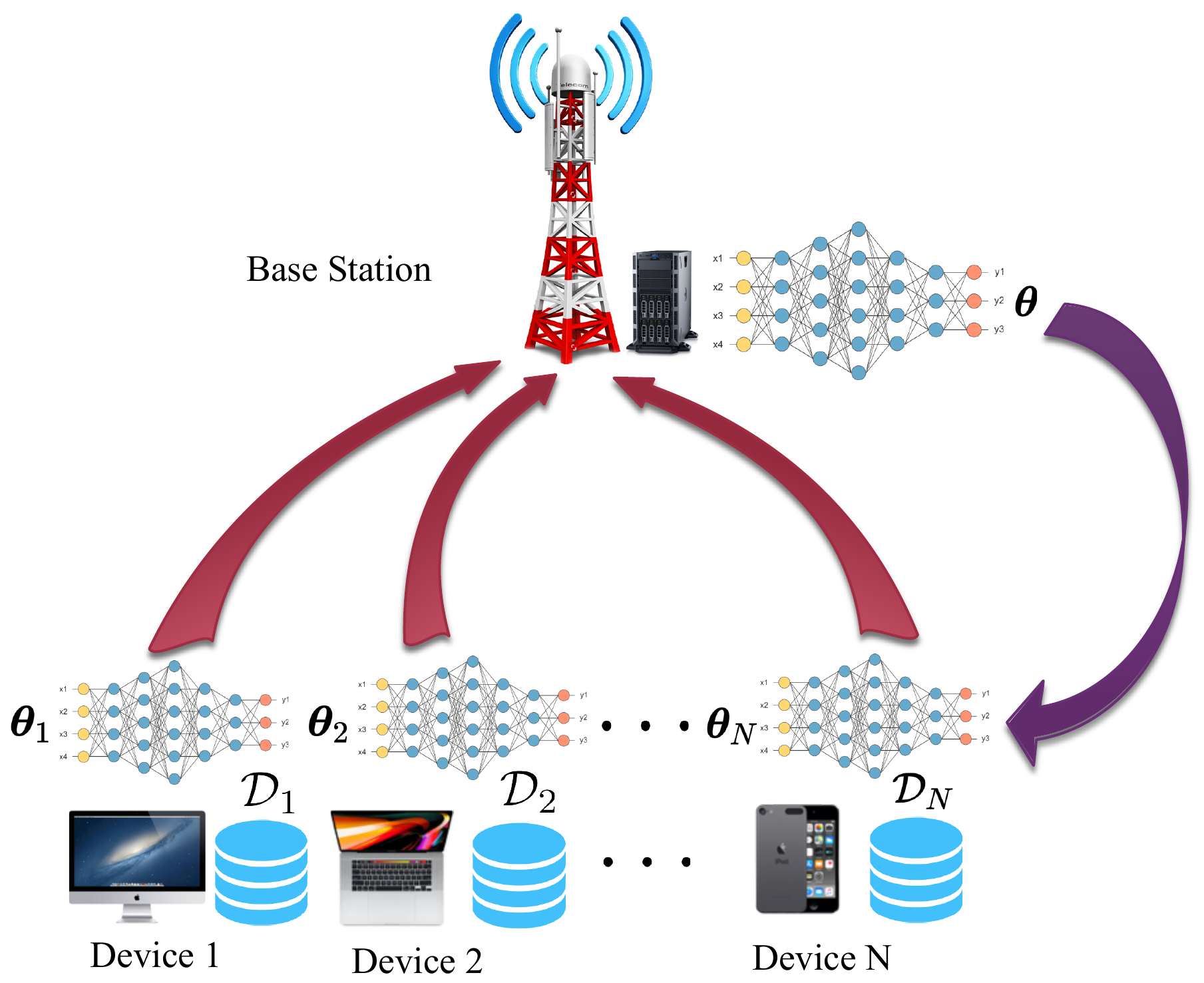}
        \caption{Illustration of the wireless federated edge learning system consisting of $N$ edge devices and one edge server.}
        \label{fig: system model}
        \vspace{-0.7cm}
\end{figure}

\subsection{FedSplit}
To solve problem \eqref{eq: problem}, we adopt the Fedsplit algorithm proposed in \cite{2020arXiv200505238P}, which is based on the operator splitting procedures. As depicted in Fig. \ref{fig: system model}, at each $t$-th communication round, the edge server broadcasts the current global model $\bm{\theta}^t$ to all edge devices via the downlink channel through the digital communication. Hence, the power constraint of the edge server is not as strict as the edge devices, and the downlink communication is assumed to be error-free \cite{9076343, 9042352, 9014530}. Consequently, each device receives the current shared global model $\bm{\theta}^t$ without distortion. Based on the received current global model $\bm{\theta}^t$ and the local dataset $\mathcal{D}_n$, device $n$ updates the local model $\bm{\theta}^{t+1}_n$ with two steps as follows.
\begin{enumerate}
        \item Local prox step:
        \begin{align}\label{eq: local prox step}
                \bm{\theta}^{t+1/2}_n \triangleq \operatorname{prox}_{s,n}(2\bm{\theta}^t - \bm{\theta}^{t}_n),
        \end{align}
        \item Local centering step:
        \begin{align}\label{eq: local centering step}
                \bm{\theta}^{t+1}_n \triangleq \bm{\theta}^{t}_n + 2(\bm{\theta}^{t+1/2}_n  - \bm{\theta}^t ),
        \end{align}
\end{enumerate}
where the proximal operator $\operatorname{prox}_{s,n}$ is defined by 
\begin{align}\label{eq: proximal operator}
        \operatorname{prox}_{s,n}(\bm{z}) \triangleq \underset{\bm{x}\in\mathbb{R}^d}{\operatorname{arg~min}}~ f_n(\bm{x}) + \frac{1}{2s}\|\bm{z-x}\|^2_2,
\end{align}
for some step size $s>0$.

After updating local models, the edge devices transmit a function of the local model over a wireless fading multiple access channel (MAC) to the edge server. By exploiting over-the-air computation (AirComp), the edge server aggregates all devices local models in one channel use \cite{8952884}, which significantly reduces communication latency. Based on the received signals, the edge server is able to obtain an estimate $\bm{\hat{\theta}}^{t+1}$ of the average global model
\begin{align}\label{eq: average global model}
        \bm{\theta}^{t+1} \triangleq \frac{1}{N} \sum_{n=1}^N \bm{\theta}^{t+1}_n.
\end{align}
The whole procedure will continue until meeting a convergence condition. 

Consider the case when each local loss function $f_n$ is both $\ell_n$-strongly convex and $L_n$-smooth. We define the condition number of the problem $\kappa = \frac{L^*}{\ell_*}$ where $\ell_* = \min _{n\in N} \ell_{n}$ is the smallest strong convexity constant and $L^* = \max _{n\in N} L_{n}$ is the largest Lipschitz constant. Then we have the following results, which has been proved in \cite[Section 5]{2020arXiv200505238P}.
\begin{thm}\label{thm: theorem 1}
Assuming that $\{\bm{\theta}_n^*\}$ are fixed points for the \textbf{Fedsplit} algorithm. Then
 for any initialization $\bm{\theta}^1 \in \mathbb{R}^d$ and step size $s=1/\sqrt{\ell_*L^*}$, 
 \begin{enumerate}
        \item the algorithm has an optimal solution $\bm{\theta}^* = \frac{1}{N} \sum_{n=1}^N\bm{\theta}_n^*$ to the problem \eqref{eq: problem};
        \item the iterates \eqref{eq: average global model} satisfy
        \begin{align*}
                \left\|\bm{\theta}^{(t+1)}-\bm{\theta}^*\right\| \leq \left(1-\frac{2}{\sqrt{\kappa}+1}\right)^{t}\sqrt{\delta_0},
        \end{align*}
        where $\delta_0 = \frac{1}{N}\sum_{n=1}^N\left\|\bm{\theta}_n^{0}-\bm{\theta}_n^*\right\|^2$;
        \item the iteration complexity is
        \begin{align*}
                T(\epsilon, \kappa)=O\left(\sqrt{\kappa} \log (1 / \epsilon)\right),
        \end{align*}
        to achieve an $\epsilon$-accurate solution, i.e., $\|\bm{\theta}^T - \bm{\theta}^*\|\leq \epsilon$.
                
 \end{enumerate}
\end{thm}

\emph{Remark 1}: To address problem \eqref{eq: problem}, a number of different methods have been proposed, e.g., FedAvg \cite{2016arXiv160205629B} and FedProx \cite{2018arXiv181206127L}. However, these methods guarantee convergence to fixed points, but not necessarily optimal solutions, even in strongly convex settings \cite{2020arXiv200505238P, 2020arXiv200401442M}. Furthermore, the iteration complexities of these algorithms will increase significantly when the problem becomes ill-conditioning. In constrast, the FedSplit algorithm enjoys global optimality and linear convergence rate when the local loss functions are strongly convex. In addition, the FedSplit algorithm is robust to the condition number of the problem. Nevertheless, the authors in \cite{2020arXiv200505238P} consider both uplink and downlink channels are error-free, which is indeed unpractical due to limited communication resources in wireless federated learning. Hence, in this paper, we study a more practical implementation of the FedSplit algorithm over a noisy fading MAC.

\subsection{Communication Protocol}
%\begin{figure}
%       \centering
%       \includegraphics[scale=0.25]{figures/transmissionProtocol.png}
%       \caption{Illustration of the transmission protocol for NOMA.}
%       \label{fig: communication protocol}
%\end{figure}

In this paper, all edge devices communicate with the edge server over the shared wireless MAC channel via AirComp. In this case, coding is unnecessary for an AirComp system to achieve the optimal tradeoff between computation rate and accuracy \cite{4305404}. Hence, this paper adopts an uncoded nonorthogonal multiple access (NOMA) protocol.
% as illustrated in Fig. \ref{fig: communication protocol}. 
 Under this setup, we assume a block flat-fading channel, where the channel coefficient remains same during one communication block. Each block is assumed to contain $d$ time slots, so that a $d$-dimensional local model is allowed to transmit within one block. Due to limited memory and computational capacity, the models on edge devices are usually tiny, and may consist of only thousands of parameters \cite{ravi2019efficient}. Since typical coherence blocks also have the same order of magnitude \cite{8608340}, it is possible to transmit a model vector in one transmission block. For large model dimensions, we can transmit the models during multiple consecutive coherence blocks, which will slightly affect the analysis. Hence, this paper mainly focuses on the former case.

In this paper, the number of blocks is assumed to be equal to the number of iterations, so that all devices upload their local models in the $t$-th iteration corresponding to the $t$-th block. Then, the received signal at the edge server is given by
\begin{align}\label{eq: MAC channel}
        \bm{y}^t = \sum_{n=1}^N h_n^t\bm{x}_n^t + \bm{w}^t,
\end{align}
where $h_n^t \in \mathbb{C}$ is the channel coefficient for device $n$ in the $t$-th block;  $\bm{w}_n^t \in \mathbb{C}^d$ denotes the additive noise i.i.d according to $\mathcal{CN}(0, \sigma^2_w\mathbf{I})$; and transmitted signal $\bm{x}_n^t$ encodes the information about the local model $\bm{\theta}^t_n$. In addition, the transmit power constraint of each device is given by 
\begin{align}\label{eq: power constraint}
        \mathbb{E}\left[\|\bm{x}_n^t\|^2_2\right] \leq P_0,\quad \forall n \in \mathcal{N},
\end{align}
with $P_0$ as the maximum transmit power.

Based on the received signal $\bm{y}^t$, the edge server needs to recover the average global model \eqref{eq: average global model} via AirComp. However, due to the distortion and noise caused by wireless channels, the edge server can only use the perturbed information about local models received from the devices to update the global model. In addition, these factors will greatly affect the convergence of the FedSplit algorithm. Hence, this paper aims to develop a reliable transceiver strategy based on the FedSplit algorithm in wireless communication systems. The strategy includes precoding local models at edge devices and recovering the average global model at edge server. In the following, we will propose the \textbf{AirComp based Fedsplit Algorithm}, and then provide convergence analysis in section \ref{sec: convergence analysis}.

\subsection{AirComp Based FedSplit Algorithm}
In this paper, we assume that perfect CSI are available on all devices and the edge sever, which can be achieved by pilot based methods. For implementation of AirComp, each device is required to perform magnitude alignment to reduce the received signal to the desired average global model \eqref{eq: average global model}. By exploiting the knowledge of CSI, each device is able to implement channel inversion by multiplying the local model by its inverse channel coefficient. Specifically, in the $t$-th iteration, device $n$ encodes its local model $\bm{\theta}^t_n$ into the transmitted signal $\bm{x}^t_n$ via 
\begin{align}\label{eq: transmitted signal}
        \bm{x}^t_n \triangleq \sqrt{\alpha^t}\frac{(h^t_n)^\mathrm{H}}{|h^t_n|^2}\bm{\theta}_n^t,
\end{align}
where $\sqrt{\alpha^t}$ is a uniform scaling factor in the $t$-th iteration. The uniform scaling factor $\sqrt{\alpha^t}$ satisfies the power constraint \eqref{eq: power constraint}, likely, 
\begin{align}\label{eq: power constraint 1}
        \|\bm{x}_n^t\|^2_2 = \left\|\sqrt{\alpha^t}\frac{(h^t_n)^\mathrm{H}}{|h^t_n|^2}\bm{\theta}^t_n\right\|^2_2 \leq P_0, \forall n,
\end{align}
which implies $\sqrt{\alpha^t} \triangleq \min_{n\in\mathcal{N}} \frac{|h_n^t|\sqrt{P_0}}{\|\bm{\theta}_n^t\|_2}$.
However, it is obvious to note that weak channels (i.e., $|h^n_t| \approx 0$) results in the small scaling factor $\sqrt{\alpha^t}$. Consequently, the received signals will be weakened and the interference caused by channel noise will significantly increase. This suggests that uniform channel inversion may not be always desirable and the optimal power-control policy for AirComp should be adapted to multiuser CSI. Therefore, we propose a binary scheme of device selection based on multiuser CSI. In particular, a threshold $\gamma$ is set for device selection, and edge devices observing fading coefficients of a smaller magnitude than $\gamma$ do not transmit in the corresponding communication round. Under this scheme, the transmitted signals \eqref{eq: transmitted signal} become
\begin{align}\label{eq: transmitted signal 1}
        \bm{x}^t_n = \sqrt{\alpha^t}\beta^t_n\frac{(h^t_n)^\mathrm{H}}{|h^t_n|^2}\bm{\theta}_n^t,
\end{align}
where the indicator of device selection $\beta^t_n$ is defined by
\begin{align}\label{eq: device selection}
        \beta^t_n = \left\{
        \begin{array}{ll}
                0, &|h^t_n| < \gamma; \\
                1, &|h^t_n| \geq \gamma,
        \end{array}
        \right.
\end{align}
with some threshold $\gamma \geq 0$. Hence, the scaling factor becomes $\sqrt{\alpha^t} = \min_{n\in\mathcal{B}} \frac{|h_n^t|\sqrt{P_0}}{\|\bm{\theta}_n^t\|_2}$.

To simplify, we denote $\mathcal{B}^t \subseteq \mathcal{N}$ as the set of participating devices indices subject to $\beta^t_n = 1$. Since the edge server is assumed to know all CSI, it also knows $\mathcal{B}^t$ by \eqref{eq: device selection}. Hence, the average global model \eqref{eq: average global model} can be recovered by the edge server as follows, 
\begin{align} \label{eq: recover}
\begin{array}{ll}
        \hat{\bm{\theta}}^{t+1} &\triangleq \frac{1}{\sqrt{\alpha^t}|\mathcal{B}^t|}\bm{y}^t \\
        &= \frac{1}{\sqrt{\alpha^t}|\mathcal{B}^t|}\left(\sum_{n=1}^N h_n^t\sqrt{\alpha^t}\beta^t_n\bm{\theta}_n^t+ \bm{w}^t\right) \\
        &= \frac{1}{|\mathcal{B}^t|}\sum_{n\in\mathcal{B}^t} \bm{\theta}_n^t + \frac{\bm{w}^t}{\sqrt{\alpha^t}|\mathcal{B}^t|}\\
        &= \frac{1}{|\mathcal{B}^t|}\sum_{n\in\mathcal{B}^t} \bm{\theta}_n^t + \tilde{\bm{w}}^t \\
        &= \bar{\bm{\theta}}^t + \tilde{\bm{w}}^t,
\end{array}
\end{align}
where $\bar{\bm{\theta}}^t \triangleq \frac{1}{|\mathcal{B}^t|}\sum_{n\in\mathcal{B}^t} \bm{\theta}_n^t$ and $\bm{\tilde{w}}^t$ is the equivalent additive noise according to $\mathcal{CN}(0, \frac{\sigma^2_w}{\alpha^t|\mathcal{B}^t|^2}\mathbf{I})$. The resulting algorithm with $T$ communication rounds is summarized in Algorithm \ref{algorithm 1}.

% pseudocode
\begin{algorithm}[tbp]
  \caption{AirComp based FedSplit Algorithm} 
  \label{algorithm 1}
%       \DontPrintSemicolon
        \SetAlgoLined
%       \KwData{local datasets $\{\mathcal{D}_n\}$}
        \SetKwInOut{Input}{Input}
        \SetKwInOut{Output}{Output}
        \SetKwFor{ParFor}{for each}{do in parallel}{end}
        \Input{Initial $\bm{\theta}^0$, threshold $\gamma$, step size $s$,  max number of rounds $T$}
%       \Output{$A$, $B$}
        Initialization for each device: $\bm{\theta}^0_n = \bm{\theta}^0, n \in \mathcal{N} $ with the initial $\bm{\theta}^0$\;
        \For{$t=0,1,\ldots,T$}{
        All devices receive the current estimate $\bm{\hat{\theta}}^t$ \;
        \ParFor{$n \in \mathcal{N}$}{
        Updating $\bm{\theta}^{t+1}_n$ via \eqref{eq: local prox step} and \eqref{eq: local centering step} \;
        Checking the channel state $h^t_n$ and determining $\beta^t_n$ via \eqref{eq: device selection} \;
        \If{$\beta^t_n = 1$}{
        Transmitting $\bm{x}^t_n$ encoded via \eqref{eq: transmitted signal 1} over the MAC \eqref{eq: MAC channel}\;
        }
        }
        The edge server receives $\bm{y}^t$, recovers $\bm{\theta}^{t+1}$ via \eqref{eq: recover}, and then broadcasts $\bm{\hat{\theta}}^{t+1}$ back to all the devices via an error-free channel \;
        }
\end{algorithm}

\section{Convergence Analysis} \label{sec: convergence analysis}
In this section, we will provide the convergence analysis of the AirComp based Fedsplit algorithm and prove that it can converge to the global optimal solution under strongly convex and smooth local loss functions. 

The main strategy of our proof is to introduce two sequences $\{\bm{\theta}^t\}$ and $\{\hat{\bm{\theta}}^t\}$ generated by \eqref{eq: average global model} and \eqref{eq: recover}, respectively. While the sequence $\{\hat{\bm{\theta}}^t\}$ is perturbed by the channel gain and noise, we still can establish a single step recursive bound for the error $\mathbb{E}\|\hat{\bm{\theta}}^t-\bm{\theta}^t\|^2$. Then by exploiting the results of Theorem \ref{thm: theorem 1}, we further characterize the convergence of $\{\bm{\theta}^t\}$. Before presenting our main results, we first have the following assumptions that our analysis is based on.

\subsection{Preliminaries}
\begin{ass}\label{thm: assumption 1}
        The local loss function $f_n$ is both $\ell_n$-strongly convex and $L_n$-smooth for any $n \in \mathcal{N}$.
\end{ass}

\begin{ass}\label{thm: assumption 2}
        The local model is bounded by a universal constant $G>0$, likely, $\|\bm{\theta}^t_n\|^2 \leq G^2, \forall t, n$.
\end{ass}

\begin{ass}\label{thm: assumption 3}
        At each communication iteration, the set of participating devices $B_t$ satisfies $|B_t|=B \leq N$ and is uniformly distributed over all the subsets of $\mathcal{N}$.
\end{ass}
Assumptions \ref{thm: assumption 1} and \ref{thm: assumption 2} are commonly used in analyzing FL algorithm for many learning-based tasks, i.e., linear regression and logistic regression. Assumption \ref{thm: assumption 3} can be implemented by the following distributed mechanism. At each communication round $t$, we choose the top $B$ devices among the participating device set $\mathcal{B}^t$ in term of their CSI, i.e., $|h^t_n|$, to transmit their signals. If the event of $|\mathcal{B}_t| < B$ happens, the devices need wait to the next communication round. Actually, the probability of the event is very small especially when $N$ is large and $\gamma$ is small. This mechanism guarantees $|\mathcal{B}_t| = B$ at each iteration. Notice that Assumption \ref{thm: assumption 3} is only used for convergence analysis. In fact, when AirComp based Fedsplit is implemented without such a mechanism, i.e., $|\mathcal{B}_t|$ is random, it will achieve similar convergence characteristics. The authors in \cite{2020arXiv200912787S} have made similar assumptions.

\subsection{Main Results}
Based on the above assumptions, we present the convergence of the sequence $\{\hat{\bm{\theta}}^t\}$ as follows.
\begin{thm}\label{thm: theorem 2}
        Consider the system model specified in Section \ref{sec: system model}. Let $\bm{\theta}^*$ denote the solution of the optimization problem \eqref{eq: problem}. When Assumptions (\ref{thm: assumption 1}-\ref{thm: assumption 3}) holds and setting the step size $s=\frac{1}{\sqrt{\ell_{*} L^{*}}}$, then it holds that
        \begin{align}
        \mathbb{E}[F(\hat{\bm{\theta}}^t)]  - F(\bm{\theta}^*) \leq \frac{\delta_0}{2L}\rho^{t} + \frac{G^2}{2B^2L}\left(B + \frac{d\sigma_w^2}{\gamma^2P_0}\right),
        \end{align}
        where $\delta_0 = \frac{1}{N}\sum_{n=1}^N\left\|\bm{\theta}_n^{0}-\bm{\theta}_n^*\right\|^2$, $\rho = \left(1-\frac{2}{\sqrt{\kappa}+1}\right)^2$ and $L = \sum_{n=1}^NL_n$.
\end{thm}

\begin{proof}
        See Section \ref{sec: proof of theorem 1}.
\end{proof}

\emph{Remark 2}: Theorem 2 shows that it is able to achieve the convergence rate of the FedSplit algorithm, i.e., linear convergence. Note that the iteration complexity remains $T(\epsilon, \kappa)=O\left(\sqrt{\kappa} \log (1 / \epsilon)\right)$ as claimed in Theorem \ref{thm: theorem 1}, whereas the GD based algorithms for the wireless FL problem proposed in \cite{9076343, 9042352} are linearly dependent of the condition number, i.e., $T(\epsilon, \kappa)=O\left(\kappa \log (1 / \epsilon)\right)$. Hence, we can conclude that our algorithm is more robust to the ill-conditioned problems. What's more, Theorem 2 establishes a finite bound of the estimation error for strongly convex and smooth local loss functions over fading MAC. The error bound is characterized by two terms, the \emph{initial distance} due to the error in the initial estimate and the \emph{additive noise} caused by the channel noise. In the case of error-free channels, our algorithm can be reduced to the FedSplit algorithm, and thus achieves the same performance. In addition, the design of the threshold $\gamma$ will greatly affect the additive noise term, which will be discussed in our future work.

\section{Numerical Experiments}
In this section we numerically evaluate the performance of AirComp based FedSplit Algorithm by presenting a typical example, i.e., linear regression. Local dataset of each device $\mathcal{D}_n$ is randomly generated by linear model
\begin{align}\label{eq: linear model}
        Y_{n} = X_{n}\bm\theta_{0}+v_{n}, \forall n \in \mathcal{N},
\end{align}
where $Y_{n}\in\mathbb{R}^{m_{n}}$ is a output vector with $m_{n}$ elements related to the design matrix $X_{n}$, $\bm\theta_{0} \in \mathbb{R}^d$ is generated by sampling from the standard Gaussian distribution $\mathcal{N}(0, \mathbf{I}_{d})$ and the noise vectors are independently generated according to $\mathcal{N}(0, \sigma^2 \mathbf{I}_{m_{n}})$. The details of generating $X_n$ will be discussed in Sections \ref{sec: well-conditioned} and \ref{sec: ill-conditioned}. We use a linear least square loss function for each device $n$, given by
\begin{align}\label{eq: loss function}
        f_{n}(\bm\theta) = \frac{1}{2}\left\|Y_{n} - X_{n}\bm\theta\right\|^2,
\end{align}
which is strongly convex and differentiable.

We evaluate the algorithms in term of the expected loss of the objective values by running an algorithm, i.e., $\mathbb{E}[F(\hat{\bm{\theta}}^t)] - F(\bm{\theta^{*}})$, where $\bm\theta^{*}$ is the solution to federated optimization problem \eqref{eq: problem}. In addition, we compare the AirComp based FedSplit algorithm with the Gradient Based Multiple Access (GBMA) algorithm proposed in \cite{9076343}, which is also developed to solve the federated learning problem over MAC. All the experiments will be performed $p$ times, and we take the average of the results.
%In GBMA, each edge device transmits an analog function of gradient of $d$-dimension local model and the edge server receives a superposed signal with additive noise and perturbations caused by fading channel. Then the edge server obtains an estimation and updates the global model, broadcasting the global model $\bm\theta^t$ back to all edge devices via the downlink channel.
In the following, we consider two different settings for problem conditioning.

\subsection{Well-conditioned Setting}\label{sec: well-conditioned}
In this case, we generate random matrices $X_n \in \mathbb{R}^{m_{n} \times d}$ with $(X_n)_{uv} \overset{\mathrm{i.i.d}}{\sim} \mathcal{N}(0, 1)$, for all $n \in \mathcal{N}$, $u \in [m_{n}]$ and $v \in [d]$. The simulation parameters are set as
\begin{align}\label{eq: Well-conditioned Setting}
\begin{array}{lll}
        p = 20, &N = 100, &m_n = 200, \\
     d = 6, &\sigma^2 = 0.25, &\sigma^2_w = 1,\quad\gamma=0.5,\\
\end{array}
\end{align}
thus satisfying that all $X_{n}$ are full rank due to $m_n \gg d$ for all device $n$, which is called as well-conditioned setting.

Except GBMA, we also compare AirComp based FedSplit Algorithm with following algorithms: (i) original FedSplit algorithm; (ii) FedSGD algorithm, $e = 1$, which is the original version of GBMA without channel distortion, where $e$ is the number of local gradient steps. After defining these parameters, we run simulations according to (\ref{eq: linear model}) and (\ref{eq: loss function}) with Rayleigh channel gain and additive noise defined before. The simulation results are shown in Fig. \ref{fig: result1}.

As illustrated in Fig. \ref{fig: result1}, although both GBMA and AirComp based FedSplit achieve a linear convergence rate, GBMA has a larger error gap, i.e., $10^{-2}$ while AirComp based FedSplit can converge to a more accurate solution, i.e., $10^{-4}$. Besides, AirComp based FedSplit and original FedSplit have the same convergence rate at the beginning. However, due to fading channels and additive noise, there is still a boundary between the solution obtained by the former and the latter, which corresponds to our theoretical analysis.

\begin{figure}[tbp]
        \centering
        \includegraphics[width=0.8\linewidth]{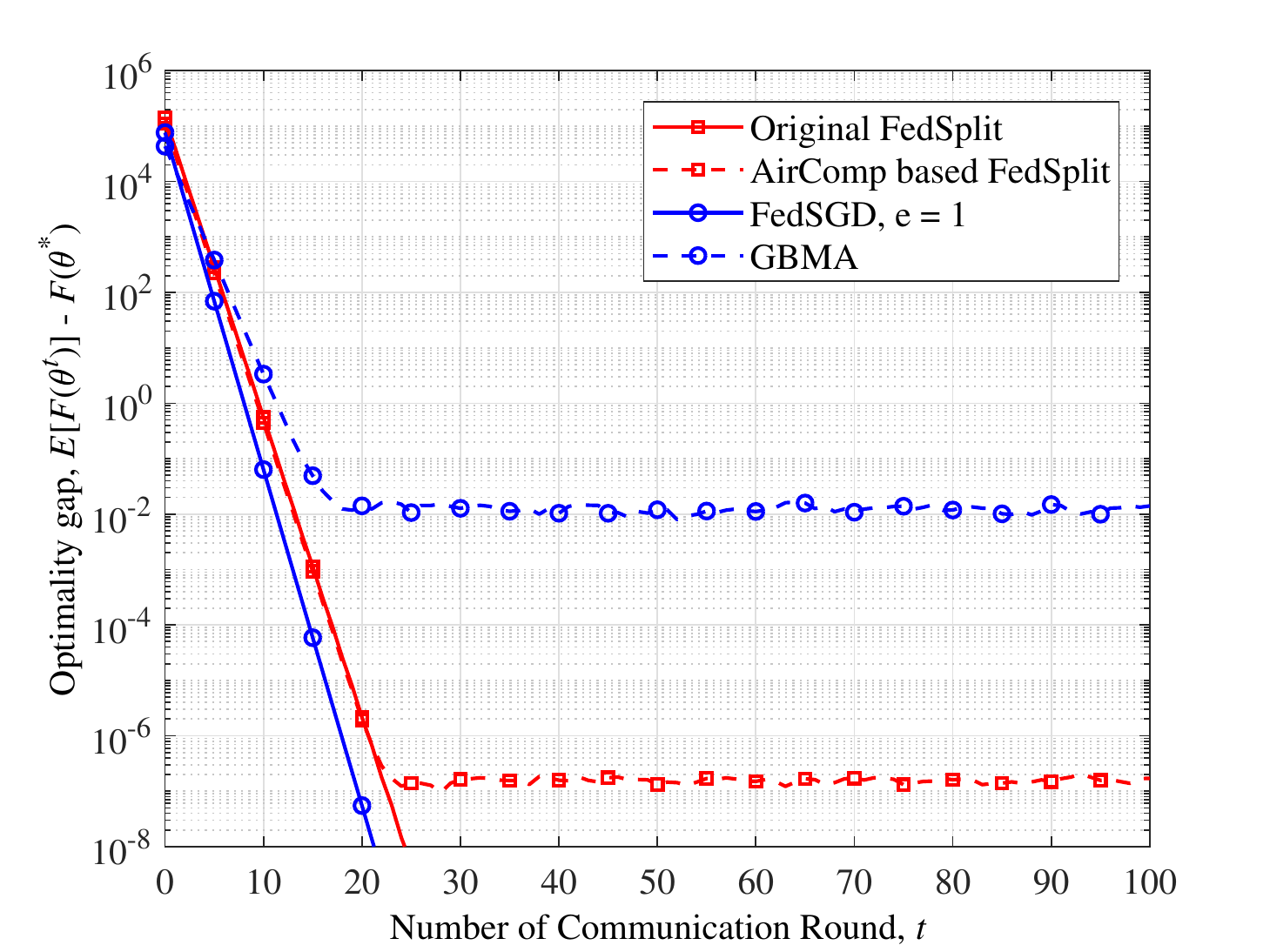}
        \caption{Simulation results for linear regression under well-conditioned setting, plotting log optimality gap versus number of communication round $t$.}
        \label{fig: result1}
        \vspace{-0.6cm}
\end{figure}

\begin{figure}[tbp]
        \centering
        \includegraphics[width=0.8\linewidth]{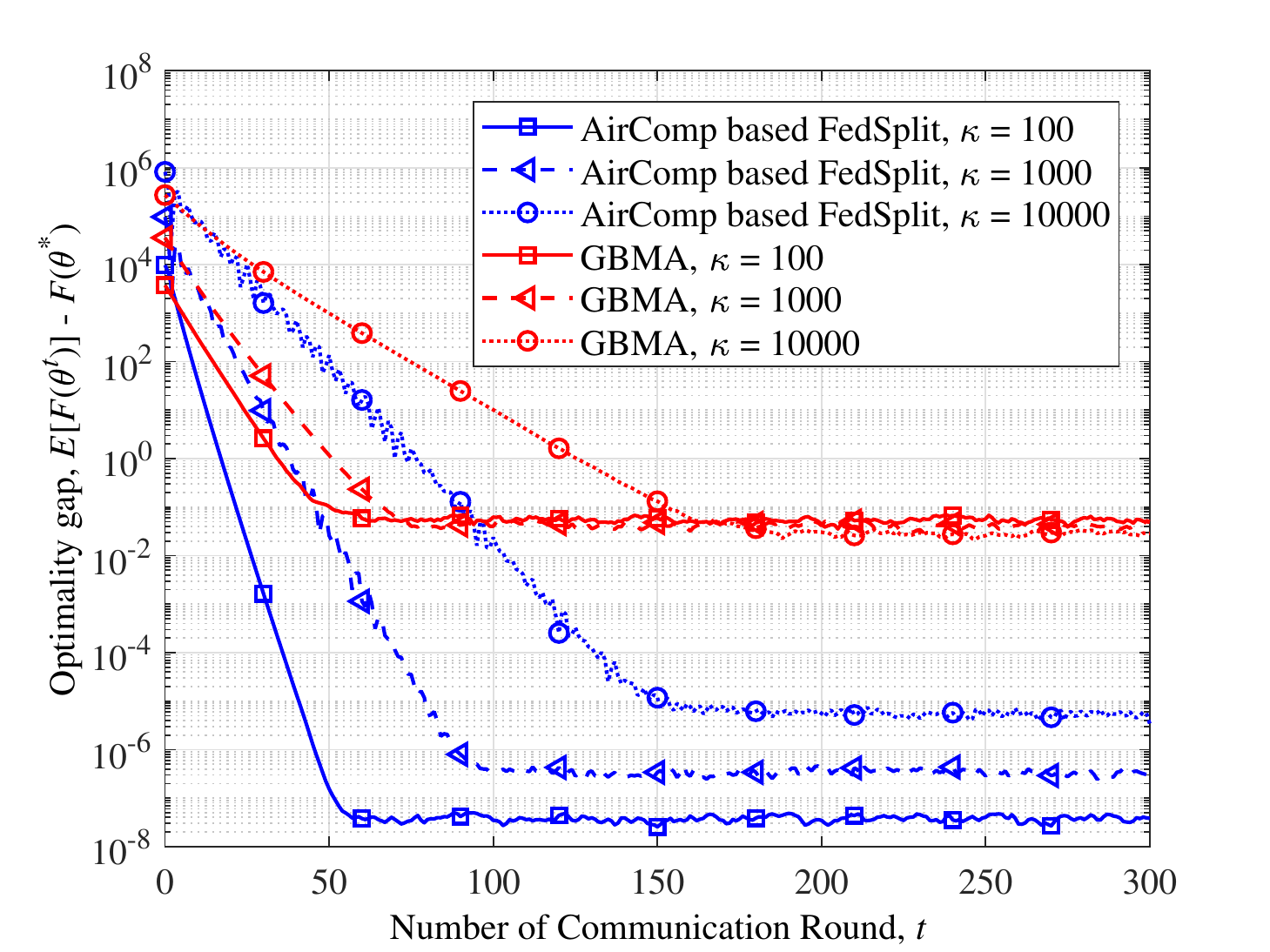}
        \caption{Simulation results for linear regression under ill-conditioned setting, plotting log optimality gap versus number of communication round $t$.}
        \label{fig: result2}
        \vspace{-0.6cm}
\end{figure}

\subsection{Ill-conditioned Setting}\label{sec: ill-conditioned}
%In this section, we provide numerical instances to illustrate the performance of AirComp based Fedsplit Algorithm versus GBMA . 
%
%Ill-conditioned problem caused by heterogeneous scalings and dependencies between different features is very common in the field of machine learning. 

To verify the effectiveness of our algorithm under ill-conditioned setting, we consider the linear regression problem with different ill-conditioned setting in term of the condition number $\kappa$. The detail of generating the design matrix $X_n$ under different condition numbers can be found in \cite[Section 4.3]{2020arXiv200505238P}.
Similarly, we use i.i.d Rayleigh channel gains and additive Gaussian noise for each device to simulate real situation. The other parameters are set as follows,
\begin{align}\label{eq: Ill-conditioned Setting}
\begin{array}{llll}
        p = 20, &N = 100, &m_n = 200, &\\
     d = 6, &\sigma^2 = 1, &\sigma^2_w = 1, &\gamma=0.5.
\end{array}
\end{align}

By running experiments for $\kappa \in \mathcal{K} = \{10^2,10^3,10^4\}$, we are able to plot the log gaps of different $\kappa$ from the two algorithms versus communication round $t$. Fig. \ref{fig: result2} shows that all the algorithms can achieve linear convergence rate. Besides, when the condition number $\kappa$ increases, the convergence will slow down. In particular, the convergence rate of AirComp based Fedsplit is less sensitive than GBMA in ill-condition cases, which means AirComp based Fedsplit is more robust to the condition number than GBMA. In other words, for ill-conditioned problems and wireless environment, AirComp based Fedsplit is faster and able to achieve higher accuracy compared with GBMA.

%\vspace{-0.3cm}
\section{Conclusion}
In this paper, we studied a wireless FL problem over a noisy fading MAC. To tackle the performance degradation in ill-conditioned settings, we proposed the AirComp based FedSplit algorithm, where the edge server recovered the noisy aggregation of local models transmitted by the end devices via AirComp. We provided the convergence analysis for the proposed algorithm that linearly converges to the optimal solutions for strongly convex and smooth loss functions. The robustness of the proposed algorithm to ill-conditioned problems with fast convergence was verified by theoretical results and numerical experiments.

%\vspace{-0.3cm}
\appendices
\section{Proof of Theorem \ref{thm: theorem 1}} \label{sec: proof of theorem 1}

According to Assumption 1, the objective function $F$ is also $L$-smooth with the Lipschitz constant $L = \sum_{n=1}^N L_n$. By the smoothness of the objective function $F$, we have that 
\begin{align}\label{eq: smoothness of F}
        \mathbb{E}[F(\hat{\bm{\theta}}^t)]  - F(\bm{\theta}^*) \leq \frac{L}{2}\mathbb{E}\left[\|\hat{\bm{\theta}}^t-\bm{\theta}^*\|^2\right].
\end{align}
By introducing the auxiliary sequence $\{\bm{\theta}^t\}$, we can rearrange the error term $\mathbb{E}\left[\|\hat{\bm{\theta}}^t-\bm{\theta}^*\|^2\right]$ as follows
\begin{align*}
        &\mathbb{E}\left[\|\hat{\bm{\theta}}^t-\bm{\theta}^*\|^2\right]= \mathbb{E}\left[\|\hat{\bm{\theta}}^t-\bm{\theta}^t + \bm{\theta}^t-\bm{\theta}^*\|^2\right]\\
        &= \mathbb{E}\left[\|\hat{\bm{\theta}}^t-\bm{\theta}^t\|^2\right]+ 2\mathbb{E}\left[(\hat{\bm{\theta}}^t-\bm{\theta}^t)^\top(\bm{\theta}^t-\bm{\theta}^*)\right]\\ 
        &+ \mathbb{E}\left[\|\bm{\theta}^t-\bm{\theta}^*\|^2\right].
\end{align*}
According to \eqref{eq: recover} and Assumption \ref{thm: assumption 3}, it is easy to verify that $\hat{\bm{\theta}}^t$ is an unbiased estimator of $\bm{\theta}^t$ as follows,
\begin{align*}
        \mathbb{E}[\hat{\bm{\theta}}^t] &= \mathbb{E}\left[\frac{1}{B}\sum_{n\in\mathcal{B}^t}\bm{\theta}_n^t + \tilde{\bm{w}}\right]
        = \mathbb{E}\left[\frac{1}{B}\sum_{n\in\mathcal{B}^t}\bm{\theta}_n^t\right] + \mathbb{E}[\tilde{\bm{w}}]\nonumber \\
        &\overset{(\mathrm{i})}{=} \frac{1}{B}\mathbb{E}\left[\sum_{n\in\mathcal{B}^t}\bm{\theta}^t_n\right] \overset{(\mathrm{ii})}{=} \frac{1}{N}\sum_{n=1}^N \bm{\theta}^t_n = \bm{\theta}^t,
\end{align*}
where $(\mathrm{i})$ comes from that $\tilde{\bm{w}}^t$ is zero-mean and $(\mathrm{ii})$ can be easily derived from \cite[Lemma 4]{li2019convergence}. Hence, 
\begin{align*}
\mathbb{E}\left[(\hat{\bm{\theta}}^t-\bm{\theta}^t)^\top(\bm{\theta}^t-\bm{\theta}^*)\right] = 0,
\end{align*}
which further implies 
\begin{align*}
        \mathbb{E}\left[\|\hat{\bm{\theta}}^t-\bm{\theta}^*\|^2\right] = 
        \mathbb{E}\left[\|\hat{\bm{\theta}}^t-\bm{\theta}^t\|^2\right] + 
        \mathbb{E}\left[\|\bm{\theta}^t-\bm{\theta}^*\|^2\right].
\end{align*}
The error term $\mathbb{E}\left[\|\hat{\bm{\theta}}^t-\bm{\theta}^*\|^2\right]$ can be divided into two terms $\mathcal{A}_1 = \mathbb{E}\left[\|\hat{\bm{\theta}}^t-\bm{\theta}^t\|^2\right]$ and $\mathcal{A}_2 = \mathbb{E}\left[\|\bm{\theta}^t-\bm{\theta}^*\|^2\right]$. In the following, we establish the upper bounds for these two terms respectively in oder to bound the error term.

Recall that $\hat{\bm{\theta}}^t=\bar{\bm{\theta}}^t + \tilde{\bm{w}}^t$ in \eqref{eq: recover}, we obtain that 

\begin{align*}
        \mathcal{A}_1 &= \mathbb{E}\left[\left\|
        \bar{\bm{\theta}}^t - 
        \bm{\theta}^t 
        \right\|^2\right]+2\mathbb{E}\left[\left(\bar{\bm{\theta}}^t - 
        \bm{\theta}^t \right)^\top\tilde{\bm{w}}^t\right] +\mathbb{E}\left[\|\tilde{\bm{w}}^t\|^2\right].
\end{align*}

Since the equivalent noise $\tilde{\bm{w}}^t$ is independent of the models and zero-mean, the term $\mathbb{E}\left[\left(\bar{\bm{\theta}}^t - 
        \bm{\theta}^t \right)^\top\tilde{\bm{w}}^t\right]$ is zero. In addition, the noise also satisfies 
        \begin{align}\label{eq: expected norm of w}
                \mathbb{E}\left[\|\tilde{\bm{w}}^t\|^2\right] = \frac{d\sigma^2_w}{\alpha^tB^2} \leq \frac{d\sigma_w^2G^2}{\gamma^2B^2P_0},
        \end{align}
        where the last inequality comes from that the definition of $\alpha_t$ and Assumption \ref{thm: assumption 2}.

        To further bound the term $\mathbb{E}\left[\left\|
        \bar{\bm{\theta}}^t - 
        \bm{\theta}^t 
        \right\|^2\right]$, we refer to the proof of \cite[Lemma 5]{li2019convergence} and Assumption \ref{thm: assumption 2}. Then, 
        \begin{align}\label{eq: expected norm of diff}
                &\mathbb{E}\left[\left\|\bar{\bm{\theta}}^t-\bm{\theta}^t\right\|^2\right] 
                = \mathbb{E}\left[\left\|\frac{1}{B}\sum_{n\in\mathcal{B}^t}\bm{\theta}^t_n-\bm{\theta}^t\right\|^2\right] \nonumber \\
                &= \frac{1}{B^2} \mathbb{E}\left[\left\|\sum_{i=1}^B(\bm{\theta}^t_{n_i}-\bm{\theta}^t)\right\|^2\right] = \frac{1}{B^2} \sum_{i=1}^B\mathbb{E}\left[\left\|\bm{\theta}^t_{n_i}-\bm{\theta}^t\right\|^2\right] \nonumber \\
                &= \frac{1}{B} \sum_{n=1}^N \frac{1}{N} \left\|\bm{\theta}^t_n-\bm{\theta}^t\right\|^2 \nonumber \\
                &= \frac{1}{B} \sum_{n=1}^N \frac{1}{N} \left(\|\bm{\theta}^t_n\|^2 - 2(\bm{\theta}^t_n)^\top\bm{\theta}^t + \|\bm{\theta}^t\|^2\right) \nonumber \\
                &= \frac{1}{BN} \sum_{n=1}^N \|\bm{\theta}^t_n\|^2 - \frac{1}{B}\|\bm{\theta}^t\|^2 
                \leq \frac{1}{BN} \sum_{n=1}^N \|\bm{\theta}^t_n\|^2 \leq \frac{1}{B}G^2.
        \end{align}
        Substitute \eqref{eq: expected norm of w} and \eqref{eq: expected norm of diff} into $\mathcal{A}_1$, we arrive at 
        \begin{align}
                \mathcal{A}_1 \leq  \frac{1}{B}G^2 + \frac{d\sigma_w^2G^2}{\gamma^2B^2P_0} = \frac{G^2}{B^2}\left(B + \frac{d\sigma_w^2}{\gamma^2P_0}\right).
        \end{align}
        As for the term $\mathcal{A}_2$, we exploit the result of \cite[Theorem 3]{2020arXiv200505238P} as follows
\begin{align*}
        \mathcal{A}_2 \leq \rho^{t}  \frac{1}{N}\sum_{n=1}^N\left\|\bm{\theta}_n^{0}-\bm{\theta}_n^*\right\|^2,
\end{align*}
where $\rho = \left(1-\frac{2}{\sqrt{\kappa}+1}\right)^2$. 
Combining the bounds of $\mathcal{A}_1$ and $\mathcal{A}_2$ with \eqref{eq: smoothness of F} yields the stated claim.

%\section*{Acknowledgment}

%\pagebreak
\bibliographystyle{IEEEtran}
\bibliography{reference}
\end{document}